\documentclass[a4paper]{article}

\usepackage[english]{babel}
\usepackage[utf8x]{inputenc}
\usepackage[T1]{fontenc}

\usepackage[a4paper, top=0.3in, margin=0.95in]{geometry}

\usepackage{amsthm,amsmath,amsfonts,amssymb}
\usepackage{mathtools}
\usepackage{enumitem}
\usepackage{braket} 

\usepackage[colorlinks=true,linkcolor=blue,urlcolor=blue,citecolor=blue,anchorcolor=green,pdfusetitle]{hyperref}
\usepackage{dsfont,xspace,xparse}
\usepackage{lmodern,newtxtt,microtype}
\usepackage{tabularx,xifthen,afterpage,booktabs}
\usepackage{qcircuit,algpseudocode,algorithm}

\usepackage{bbm} 
\usepackage{xifthen} 
\usepackage{authblk}
\usepackage{graphicx, caption}

\usepackage{upgreek}

\usepackage[dvipsnames]{xcolor}
\usepackage{url}
\usepackage{tikz} 

\theoremstyle{plain}
\newtheorem{thm}{Theorem}[section] 

\newtheorem{proposition}[thm]{Proposition}

\newcommand{\R}{{\mathbb{R}}}
\newcommand{\id}{{\mathbbm{1}}} 

\newcommand{\Exp}[2][]{\mathbb{E}{%
	\ifthenelse{\isempty{#1}}{}{_{#1}}}%
	\left[{#2}\right]} 
\newcommand{\cleansum}{\displaystyle \sum} 

\newcommand{\norm}[1]{\|{#1}\|}
\newcommand{\Norm}[1]{\left\|{#1}\right\|}
\newcommand{\Abs}[1]{\left|{#1}\right|}

\newcommand{\ketbra}[2]{\ket{#1}\!\!\bra{#2}}
\newcommand{\proj}[1]{\ketbra{{#1}}{{#1}}}

\renewcommand{\O}{\mathcal{O}}
\newcommand{\curly}[1]{\mathcal{{#1}}} 
\newcommand{\polylog}{\text{poly}\log}

\newcommand\blfootnote[1]{%
	\begingroup
	\renewcommand\thefootnote{}\footnote{#1}%
	\addtocounter{footnote}{-1}%
	\endgroup
}

\title{Spectral sparsification of matrix inputs as a preprocessing step for quantum algorithms}
\author[1]{Sathyawageeswar Subramanian}
\author[2,3]{Steven Herbert}
\affil[3]{\small{\textit{Department of Applied Mathematics and Theoretical Physics, University of Cambridge,  UK}}}
\affil[2]{\small{\textit{Department of Computer Science, University of Oxford, UK}}}
\affil[3]{\small{\textit{Department of Computer Science and Technology, University of Cambridge,  UK}}}
\date{}

\begin{document}
\maketitle

\begin{abstract}
\noindent We study the potential utility of classical techniques of spectral sparsification of graphs as a preprocessing step for digital quantum algorithms, in particular, for Hamiltonian simulation. Our results indicate that spectral sparsification of a graph with $n$ nodes through a sampling method, e.g.\ as in \cite{Spielman2011resistances} using effective resistances, gives, with high probability, a locally computable matrix $\tilde H$ with row sparsity at most $\O(\text{poly}\log n)$. For a symmetric matrix $H$ of size $n$ with $m$ non-zero entries, a one-time classical runtime overhead of $\O(m\norm{H}t\log n/\epsilon)$ expended in spectral sparsification is then found to be useful as a way to obtain a sparse matrix $\tilde H$ that can be used to approximate time evolution $e^{itH}$ under the Hamiltonian $H$ to precision $\epsilon$. Once such a sparsifier is obtained, it could be used with a variety of quantum algorithms in the query model that make crucial use of row sparsity. We focus on the case of efficient quantum algorithms for sparse Hamiltonian simulation, since Hamiltonian simulation underlies, as a key subroutine, several quantum algorithms, including quantum phase estimation and recent ones for linear algebra. Finally, we also give two simple quantum algorithms to estimate the row sparsity of an input matrix, which achieve a query complexity of $\O(n^{3/2})$ as opposed to $\O(n^2)$ that would be required by any classical algorithm for the task.
\end{abstract}

\section{Introduction}

In\blfootnote{$^1$\href{mailto:ss2310@cam.ac.uk}{\texttt{ss2310@cam.ac.uk}}, $ ^{2,3}$\href{mailto:sjh227@cam.ac.uk}{\texttt{sjh227@cam.ac.uk}}} classical graph theory and signal processing, sparsifying dense matrices and performing algorithms thereon to reduce the computational load is a key idea, which has received significant attention over the past several years. For instance, in the case where the matrix is a graph adjacency matrix, one of the primary motivations for spectral sparsification of the graph Laplacian in classical computer science is its utility for attacking cut problems \cite{Batson2013sparsification,Batson2014Ramanujan}, and as a result the spectral properties of the sparsified graph's Laplacian are frequently treated as the measure of interest to retain while reducing the number of edges in the graph. Sparsification has also found use in designing preconditioners for linear system solvers \cite{Spielman2011resistances,Spielman2011,Spielmanm2014}, and in studying constraint satisfaction problems \cite{Dinur2006}.\\
\indent Recently, the idea of transferring the wealth of results on the benefits of sparsity from classical signal processing to quantum algorithms, in particular with regards to Hamiltonian simulation, has been investigated in \cite{Aharonov2018sparsification}; their results for analog simulation indicate that degree reduction and edge dilution does not work in the quantum setting for general graphs. On the other hand, digital quantum simulation may also benefit from sparsification, and this idea is yet to be explored to the best of our knowledge. For Hamiltonian simulation, we are ultimately interested in approximating the evolution $U(H,t)=e^{-iHt}$ under a Hermitian matrix $H$. That is, given any state $\ket{\psi}$, we want to approximate the state $U(H,t)\ket{\psi}$, which can be done by approximating $U$ itself in spectral norm. We should then study and bound
\begin{equation}
    \sup_{\ket{\psi}\in\curly H}\Norm{U(H,t)\ket{\psi}-U(H',t)\ket{\psi}},
\end{equation}
where $H'$ is a spectral sparsifier for $H$. This quatity is simply the spectral norm of the difference operator $e^{-iHt}-e^{-iH't}$, but the form above reminds us that we also need to keep track of the way the eigenspaces change in the process of sparsification. \\
\indent In this paper we address three issues with transferring the classical sparsity results to the quantum setting directly: firstly, the classical sparsification techniques typically yield matrices that are sparse overall, but the quantum algorithm for Hamiltonian simulation requires the more restrictive condition of row-sparsity, that is quantum algorithms for Hamiltonian simulation in the query model, in which the input Hamiltonian is accessed via a unitary oracle that computes matrix entries in place, up to fixed precision, typically require the interaction graph of the Hamiltonian matrix to be row sparse and locally computable; secondly, it is necessary that the adjacency matrix, rather than Laplacian is convergent; and finally, it is necessary that any residual error in the sparsified matrix does not cause a catastrophic break-down in the accuracy of the Hamiltonian simulation.\\
\indent We show how each of these problems can be overcome, although we do still make some assumptions, namely that the Hamiltonian is real, that each row of the Hamiltonian is sampled sufficiently frequently in the sparsification method, and that the sparsified Hamiltonian commutes with actual Hamiltonian. In fact, we do not believe that any of these assumptions are necessary, rather that they are simply features of our proofs and that future analysis should reveal that the results hold when these are not made. Nevertheless, even with these assumptions, the analysis in this paper serves the purpose of better connecting classical sparsification with Hamiltonian simulation methods that assume row-sparse matrices, and shows the likely way forward to achieving a full general and unrestricted result.\\
%
%
%
%
%
%
\indent A secondary purpose of this paper concerns the verification of row-sparsity, which is also an important question in its own right, and one which ostensibly would appear to lend itself to being sped up by a quantum algorithm. We show this is indeed the case by proposing two quantum algorithms that can decide whether a given matrix is row-sparse with fewer operations than are required classically. Whilst conceptualised from quite different starting points, both of these quantum algorithms require $\mathcal{O}(n^{3/2})$ operations, compared to $\mathcal{O}(n^2)$ classically (for a $n\times n$ matrix), and this coincidence of computational complexity raises two intriguing possibilities: firstly, it may be possible to combine the two algorithms in some way to achieve the sparsity verification in $\mathcal{O}(n)$ operations; or conversely, it may be that $\mathcal{O}(n^{3/2})$ is a lower-bound.\\
\indent The remainder of the paper is organised as follows: in Section~\ref{setup} we give precise details of the problem we are going to solve, including analysis of the overall benefits that it brings to Hamiltonian simulation; in Section~\ref{mainsparsity} we give our main results on relating the classical sparsification algorithm to the problem of Hamiltonian simulation; in Section~\ref{maintest} we propose two quantum sparsification verification algorithms; and finally in Section~\ref{disc} we include a wide-ranging discussion covering, amongst other things, the physical meaning of a row-sparse Hamiltonian.

\subsection*{Contributions}

\begin{enumerate}
\item From the literature we note that sparsification can reduce the computational load in several linear algebraic problems, whilst maintaining accuracy. However, we note a discrepancy between classical sparsification algorithms, which typically achieve edge sparsification or dilution, and quantum algorithms for sparse Hamiltonian simulation, which require row sparsity (degree reduction). To bridge this divide, we provide a necessary and sufficient condition for (general) sparsity to imply row sparsity.
    
    \item The classical sparsification method upon which we base our analysis \cite{Spielman2011resistances} provides a bound in terms of the Laplacian, whereas we require one in terms of the adjacency matrix. We therefore prove that the accuracy condition proved for the Laplacian implies that the adjacency matrix is also well-approximated by the sparsified adjacency matrix, when the above row-sparsity condition is met.
    
    \item We then show that this condition on the adjacency matrix being well-approximated \textit{is} sufficient for the Hamiltonian simulation with the sparsified matrix to well-approximate the actual case.
    \item We are also interested in the verification of sparsity, and to this end we propose two quantum algorithms that can verify whether or not a matrix is row sparse in $\mathcal{O}(n^{3/2})$ time, for an $n^2$ matrix, which represents an improvement on the $\mathcal{O}(n^2)$ time required classically.
\end{enumerate}

\section{Setup and problem statement}
\label{setup}
%
%
\indent We base our exposition on the method of spectral sparsification using effective resistance sampling developed in \cite{Spielman2011resistances}. Given a graph $G=(V,E,W)$ with $|V|=n$ vertices, $|E|=m\leq \binom{n}{2}$ edges, and $W$ an $m\times m$ diagonal matrix of edge weights, spectral sparsification generates another graph $G'=(V,E',W')$ such that 
\begin{align}
\label{eq:spectral_approx}
    |E'| \in~ & \O(\frac{n\log n}{\epsilon^2}), \nonumber \\
    (1-\epsilon)L_G &\preceq L_{G'} \preceq (1+\epsilon)L_G,
\end{align}
where $L$ represents the Laplacian matrix, and the second condition holds in the usual partial order on positive semidefinite matrices. The runtime of this classical algorithm is $\O(\frac{m\log n}{\epsilon^2})$.

The simplest process of this kind can be described as sampling from the edge set $E$ according to some probability mass function (pmf) $\curly P$ and hence populating $E'$. If for each $e\in E$ the probability of picking $e$ is $p_e=\curly P(e)$, we can marginalise out the edges and consider the pmf induced on the vertices, defining $p_v = \cleansum_{e\in E : B(e, v)\neq0}p_e$ where $B$ denotes the edge-vertex incidence matrix of the graph.

In this article, we show that the sparsifier generated by a spectral sparsification algorithm which uses sampling methods has under certain conditions, with high probability, the additional property that it is row-sparse, i.e., that the maximum degree of the sparsifier grows only as $\O(\polylog(n))$.\\



\section{Relating sparsification to efficient Hamiltonian simulation}
\label{mainsparsity}


In this section we prove three results that allow us to make the connection between classical sparsification and efficient Hamiltonian simulation in the query complexity model: firstly, we prove a necessary and sufficient condition for sparsity to imply row sparsity; secondly, we show that the asymptotic convergence of the Laplacian proven by \cite{Spielman2011resistances} is sufficient for the asymptotic convergence of the adjacency matrices that we require; and finally we show that the propagation of error from the sparsification process into the sparsified Hamiltonian in the sense described in the sparsification bounds does not lead to a significant increase in error in the Hamiltonian simulation.


\subsection{A necessary and sufficient condition for row sparsity}

A matrix is row-sparse if the number of non-zero entries in a row (for all rows) grows as poly-log of the size of the matrix (at most). Clearly row-sparsity implies sparsity, but the opposite does not apply in general, for example a star network has a sparse adjacency matrix, but each element of the row corresponding to the point of the star is one (except for the leading diagonal entry), and therefore is obviously not row sparse.\\
\indent In the effective resistance Hamiltonian sparsification method, each edge of a adjacency graph corresponding to the Hamiltonian is chosen with a certain probability, and a defined number of i.i.d. samples are drawn (that grows at most with $n (\textnormal{poly} \log(n))$), to construct the sparsified Hamiltonian. Marginalising over the the edges incident to each node as described in Section~\ref{setup}, to get a vertex selection probability distribution\footnote{Note that each edge is connected to two vertices, vertex selection is not mutually exclusive so this will sum to more than one -- a property later used in upper-bounding.}, it is obvious that a necessary condition for row-sparsity is that no vertex is selected with probability growing faster than $(\textnormal{poly} \log(n) )/n$. We now show that this necessary condition is also sufficient, in both obvious senses of asymptotic statistical row sparsity:
\begin{enumerate}
    \item As $n \to \infty$ the probability that any row has a number of elements greater than $\O(\log^{b'} n)$ tends to zero, for some $b'$.
    \item As $n \to \infty$ the expected maximum number of non-zero elements (across the $n$ rows of the matrix) grows as $\mathcal{O}( \log^{b}n)$ for some $b$.
\end{enumerate}
Hereafter these are termed the first and second properties of asymptotic row sparsity.

\begin{proposition}
\label{prop1}
For a Hamiltonian sparsified by the effective resistance method, if no row is selected with probability greater than $\log^b(n) /n$, then it is satisfies the first row sparsity property.
\end{proposition}
\begin{proof}
Let the total number of samples drawn be $n \log^a n$, therefore the expected number of times that the row with maximum occupation is selected can be upper-bounded:
\begin{equation}
    \mu < \log^b n \log^a n = \log^c n,
\end{equation}
where $c$ is defined as $a+b$.\\
\indent Let $x_i$ be the number of non-zero entries in the $i^{th}$ row, and $X$ be number of non-zero entries in the row with greatest selection probability. By the Chernoff bound \cite{chernoff1952}, for $R \geq 6 \mu$:
\begin{equation}
    P(X \geq R) \leq 2^{-R},
\end{equation}
letting $R = \log^{c+1} n$ (the $c+1$ term is included to ensure that the $>6 \mu$ condition is met for sufficiently large $n$), and noticing that all the other nodes are less likely to be chosen and thus the Union Bound can be invoked:
\begin{align}
    P(\max(x_i) \geq R) \leq & \sum_i p(x_i \geq R) \nonumber \\
    \leq & n 2^{-R} \nonumber \\
    = & n 2^{-\log^{c+1} n},
\end{align}
where this upper-bound also uses the fact that all other rows are chosen with probability less than the row with maximum selection probability. Let $p=n 2^{-\log^{c+1} n}$:
\begin{align}
    \log_2 p & = \log_2 n - \log_2^{c+1} n \nonumber \\ 
    & = \log_2 n (1 - \log^{c} n) \nonumber \\
    & < 0
\end{align}
for sufficiently large $n$, i.e., as $n \to \infty , \, p(X \geq R) \to 0$.
\end{proof}
\begin{proposition}
For a Hamiltonian sparsified by the effective resistance method, if no row is selected with probability greater than $\log^b(n) /n$, then it is satisfies the second row sparsity property.
\end{proposition}
\begin{proof}
The proof is similar to that of Proposition~\ref{prop1}, and the same symbols are used. The proposition considers the expectation of the number of non-zero entries in the row with the most non-zero entries (note that this need not be the the row with highest selection probability). This can be upper-bounded:
\begin{align}
    \mathbb{E}(X) \leq & R \times p(X<R) +(n \log^a n ) \times p(X \geq R) \nonumber \\
    \label{eqref10}
    \leq & 1 \times \log^c n + n^2 \log^a n 2^{-\log^{c+1}n},
\end{align}
i.e., where the second term in the first line uses the fact that, if the maximally chosen row is chosen more than $R$ times, it can still only be chosen a total of $n \log^a n$ times, as this is the number of samples. The first term in the RHS of Eq.~ (\ref{eqref10}) is clearly $\in\O( \textnormal{poly} \log(n))$, letting $y$ equal the second term in the RHS of Eq.~(\ref{eqref10}), i.e.\
\begin{align}
    y = & n^2 \log^a n 2^{-\log^{c+1}n} \nonumber \\
    \implies \log y = & 2 \log n + a \log \log n - \log^{c+1} n \nonumber \\
    = & \log n (2- \log^{c+1} n) + a \log\log n,
\end{align}
clearly, as $n \to \infty, \, \log y \to - \infty$, therefore $y \in \mathcal{O}(1)$. So it follows that $\mathbb{E}(X) \in \mathcal{O}(\log^c n)$.
\end{proof}


\subsection{Laplacian convergence implies adjacency matrix convergence when row sparse}
\label{lap}

From \cite{Spielman2011resistances}, we have that a weighted graph with Laplacian $L$ is sparsfied to a graph with Laplacian $\tilde{L}$, such that the following condition is met:
\begin{equation}
\label{neweq10}
    (1-\epsilon)\mathbf{x}^T L \mathbf{x} \,\,\, \leq \,\,\, \mathbf{x}^T \tilde{L} \mathbf{x} \,\,\, \leq \,\,\, (1+\epsilon)\mathbf{x}^T L \mathbf{x},
\end{equation}
for all vectors $x$ and $1/\sqrt{n} < \epsilon \leq 1$. \\
\indent We wish to express a similar condition for the adjacency matrix, $A = D - L$, where $D$ is a diagonal matrix where each element is the sum of the elements in that row of the adjacency matrix. Our method to do this relies on the property $\mathbb{E}(D_{ii}) = \mathbb{E}(\tilde{D}_{ii})$ for all $i$ --  that is that the diagonal elements of the sparsified Laplacian are expected to be the same as those of the actual Laplacian (this can be seen in the analysis in \cite{Spielman2011resistances}). Additionally, we must assume that each row is expected to be selected at least $\log^{\alpha}n$ times, where $\alpha>1$. To an extent it is valid to criticise such a condition as unnecessarily restrictive, however we do expect the condition in Theorem~\ref{newthm} (or another very similar one) to hold even if this were not to be the case, albeit requiring a very different proof. The inclusion of this condition is therefore for reasons of exposition -- it enables us to use a similar proof to the others in this section, and suffices to demonstrate the principle. Moreover, in Section~\ref{disc} we discuss the physicality of such a restriction. With these restrictions in place, we can give the result as a theorem:
\begin{thm}
\label{newthm}
\begin{equation}
\label{neweq20}
   \mathbf{x}^T A \mathbf{x} -  \epsilon' \mathbf{x}^T \mathbf{x} \max_{\substack{i}} (D_{ii}) \,\,\, \leq \,\,\, \mathbf{x}^T \tilde{A} \mathbf{x} \,\,\, \leq\,\,\,  \mathbf{x}^T A \mathbf{x} +  \epsilon' \mathbf{x}^T \mathbf{x} \max_{\substack{i}} (D_{ii}),
\end{equation}
for some $\epsilon' \in \omega(1/ \log^{\sqrt{\alpha - 1}} n)$.
\end{thm}
\begin{proof} We start by substituting $A = D - L$ and $\tilde{A} = \tilde{D} - \tilde{L}$ into Eq.~ (\ref{neweq10}), and rearranging:
\begin{eqnarray}
    (1-\epsilon)\mathbf{x}^T L \mathbf{x}  \leq & \mathbf{x}^T \tilde{L} \mathbf{x} & \leq  (1+\epsilon)\mathbf{x}^T L \mathbf{x} \nonumber \\
    (1-\epsilon)\mathbf{x}^T (D-A) \mathbf{x}  \leq & \mathbf{x}^T (\tilde{D}-\tilde{A}) \mathbf{x} & \leq  (1+\epsilon)\mathbf{x}^T (D-A) \mathbf{x} \nonumber \\
    (1-\epsilon)\mathbf{x}^T D \mathbf{x} - (1-\epsilon)\mathbf{x}^T A \mathbf{x}  \leq & \mathbf{x}^T \tilde{D} \mathbf{x} - \mathbf{x}^T \tilde{A} \mathbf{x} & \leq  (1+\epsilon)\mathbf{x}^T D \mathbf{x} - (1+\epsilon)\mathbf{x}^T A \mathbf{x} \nonumber \\
    (1-\epsilon)\mathbf{x}^T A \mathbf{x} - (1-\epsilon)\mathbf{x}^T D \mathbf{x} \geq & \mathbf{x}^T \tilde{A} \mathbf{x} - \mathbf{x}^T \tilde{D} \mathbf{x}  & \geq  (1+\epsilon)\mathbf{x}^T A \mathbf{x} - (1+\epsilon)\mathbf{x}^T D \mathbf{x}\nonumber \\
    \mathbf{x}^T A \mathbf{x} - (1-\epsilon)\mathbf{x}^T D \mathbf{x} \geq & \mathbf{x}^T \tilde{A} \mathbf{x} - \mathbf{x}^T \tilde{D} \mathbf{x}  & \geq  \mathbf{x}^T A \mathbf{x} - (1+\epsilon)\mathbf{x}^T D \mathbf{x}\nonumber \\
    \mathbf{x}^T A \mathbf{x} + \epsilon \mathbf{x}^T D \mathbf{x} + \mathbf{x}^T (\tilde{D}-D) \mathbf{x} \geq & \mathbf{x}^T \tilde{A} \mathbf{x}   & \geq  \mathbf{x}^T A \mathbf{x} - \epsilon\mathbf{x}^T D \mathbf{x} - \mathbf{x}^T (D-\tilde{D}) \mathbf{x} \nonumber \\
    \mathbf{x}^T A \mathbf{x} + \mathbf{x}^T \mathbf{x} (\epsilon + \tilde{\epsilon}) \max_{\substack{i}} (D_{ii}) \geq & \mathbf{x}^T \tilde{A} \mathbf{x}   & \geq  \mathbf{x}^T A \mathbf{x} -\mathbf{x}^T \mathbf{x} (\epsilon + \tilde{\epsilon}) \max_{\substack{i}} (D_{ii}) ,
\end{eqnarray}
where we define $\forall i \, (1+ \tilde{\epsilon}) D_{ii} \leq \tilde{D}_{ii} \leq (1+ \tilde{\epsilon}) D_{ii}$, which we address using the Chernoff bound. Addressing the upper limit of $\tilde{D}_{ii}$, we have that:
\begin{align}
    \forall i \, P(\tilde{D}_{ii} \geq (1+ \tilde{\epsilon}) D_{ii} ) & \leq \exp\left( \frac{-\tilde{\epsilon}^2  D_{ii}}{3} \right) \nonumber \\
    & \leq \exp\left( \frac{-\tilde{\epsilon}^2 \min_{\substack{i}}( D_{ii})}{3} \right) \nonumber \\
    & = \exp\left( \frac{-\tilde{\epsilon}^2 \log^\alpha n}{3} \right),
\end{align}
using the Union bound, we have that
\begin{equation}
    P(\max_{\substack{i}} (\tilde{D}_{ii} - D_{ii}) >  \tilde{\epsilon} D_{ii} ) \leq n \exp\left( \frac{-\tilde{\epsilon}^2 \log^\alpha n}{3} \right).
\end{equation}
Likewise for the lower-bound on $\tilde{D}_{ii}$ we have that:
\begin{align}
    \forall i \, P(\tilde{D}_{ii} \leq (1 - \tilde{\epsilon}) D_{ii} ) & \leq \exp\left( \frac{-\tilde{\epsilon}^2  D_{ii}}{2} \right) \nonumber \\
    & \leq \exp\left( \frac{-\tilde{\epsilon}^2 \min_{\substack{i}}( D_{ii})}{2} \right) \nonumber \\
    & = \exp\left( \frac{-\tilde{\epsilon}^2 \log^\alpha n}{2} \right),
\end{align}
and again using the Union bound: 
\begin{equation}
    P(\max_{\substack{i}} ( D_{ii} -\tilde{D}_{ii}) >  \tilde{\epsilon} D_{ii} ) \leq n \exp\left( \frac{-\tilde{\epsilon}^2 \log^\alpha n}{2} \right).
\end{equation}
So we can see that the condition $\tilde{\epsilon} \in \omega(1/ \log^{\sqrt{\alpha - 1}} n)$ must hold in order for the closeness to hold asymptotically -- i.e., as $n \to \infty$ the probability that we have a `good' sparsifier tends to one. This dominates the condition $1/ \sqrt{n} < \epsilon \leq 1$, and so we can set $\epsilon' = \epsilon + \tilde{\epsilon}$ to complete the proof.
\end{proof}

\subsection{Error propagation: from spectral sparsification to hamiltonian simulation}
It remains to be shown that the fact that the adjacency matrix is well approximated by its sparsified version implies that the Hamiltonian simulation will also be well approximated when the sparsified version is used. To do so, we express Eq.~(\ref{neweq20}) in slightly different (but equivalent) terms, namely that we are given a spectral approximation $\tilde{H}$ of a Hamiltonian, $H$ satisfying 
\begin{equation}
\label{eq:spectral_approx_H}
    (1-\epsilon')H \preceq \tilde{H} \preceq (1+\epsilon')H.
\end{equation}
Consider the following
\begin{align}
    \Norm{\left( \tilde U - U\right)\ket \psi} & := \braket{\psi|\left( \tilde{U} - U\right)^\dagger \left( \tilde U - U\right)|\psi} \nonumber\\
    &= \braket{\psi|\id+\id-\tilde U^\dagger U-U^\dagger \tilde U|\psi},
\end{align}
where $\tilde U := U(\tilde{H},t) $ and $U := U(H,t)$ for convenience. Let us first take the simple case when $H$ and $\tilde H$ commute; then $\tilde U^\dagger U = e^{-it(H-\tilde H)}$, and so $\tilde U^\dagger U + U^\dagger\tilde U = 2\cos\left(t\Delta H\right)$
where $\Delta H := H - \tilde H $. Now we can write
\begin{align}
    \Norm{\left( \tilde U - U\right)\ket \psi} :&=2\braket{\psi|\left(\id-\cos\left(t\Delta H\right)\right) |\psi}\nonumber\\
    &=t^2\epsilon'^2\langle (\Delta H)^2\rangle_\psi+ \O\left((\epsilon't\Norm{H})^4\right)\nonumber\\
    &\leq \epsilon'^2\cdot(t\Norm{H})^2+ \O\left((\epsilon't\Norm{H})^4\right),
\end{align}
where Eq.~\eqref{eq:spectral_approx_H} gives $-\epsilon' H \preceq \Delta H \preceq \epsilon' H$, which also implies $0 \leq \Norm{(\Delta H)^2} \leq \epsilon'^2 \Norm{H^2} = \epsilon'^2 \Norm{H}^2$. In the last last line we bound $\langle H\rangle_\psi$, the expectation value of the Hamiltonian $H$ under the state $\psi$, by the maximum value the energy can take, given by the spectral norm of $H$.

When $H$ and $\tilde{H}$ do not commute, the first order Suzuki-Trotter formula suggests that $U^\dagger U = e^{-it(H-\tilde H)}$ up to an error of order $\O(t^2\norm{H})$. The error term in the sum $\tilde U^\dagger U + U^\dagger \tilde U$ is then third order in $t\norm{H}$. Thus the above analysis still holds for small times $t=\O(\epsilon'^{2/3})$. However, for longer evolution times, a more detailed error analysis is required, and we are currently investigating this.

Plugging $\tilde H$ into a Hamiltonian simulation algorithm (e.g.\ \cite{Low2017OptimalProcessing,Gilyen2018QuantumArithmetics}) results in a circuit that will approximate evolution under $\tilde H$ to some desired precision $\epsilon$ through a unitary circuit $\tilde U'$. Noting that $\Norm{U-\tilde U'}\leq \Norm{U-\tilde U} + \Norm{\tilde U-\tilde U'}$, we see that $\tilde H$ can be used to approximate Hamiltonian evolution under $H$.

\subsection*{Runtime overhead:} From the above analysis of error propagation, it is clear that choosing 
\begin{equation}
    \label{eq:eps_choice}
    \epsilon'=\O\left(\frac{\sqrt{\epsilon}}{t\Norm{H}}\right)
\end{equation}
ensures that $\Norm{U-\tilde U}\leq \epsilon$ to first order. Putting this back into the runtime expression for the spectral sparsification algorithm of \cite{Spielman2011resistances}, we can estimate the (one-time) classical runtime overhead required in order to use sparse Hamiltonian simulation for time $t$, which is given by
\begin{equation}
    \label{eq:classical_overhead}
    \O\left(\frac{mt\Norm{H}\log n}{\epsilon}\right).
\end{equation}
The presence of the spectral norm $\norm{H}$ is to be expected as it sets the energy scales for the problem; evidently this method is only useful if $\Norm{H}=\O(\text{poly}(n))$. We expect this to be true for several systems of physical significance, e.g.\ moleular Hamiltonians, which typically have $\O(n^4)$ terms in a tensor product of Pauli basis (expecting the coupling constants also to scale polynomially for most common molecules).

\section{Sparsity testing}
\label{maintest}
Testing if an input function or vector is sparse is a problem that has recently received some attention in the context of big data and machine learning algorithms  \cite{Barman18SparsityTesting,Gopalan11sparsityTesting}. Could a quantum algorithm for sparsity testing offer any advantages? Given an oracle that computes matrix entries in place, we could use a comparison oracle to flag an ancillary register as `1' wherever there is a zero entry, and then use quantum amplitude estimation on the ancilla to estimate the number of zero entries. We demonstrate two quantum algorithms below for testing row-sparsity of an input matrix.

\subsection{Sparsity testing using quantum amplitude estimation}
Given a matrix $A\in\R^{n\times n}$ (we make this assumption of being real for simplicity of exposition, and to be consistent with the previous analysis, but the following should easily generalise to complex numbers) that we can access via a unitary quantum oracle that computes its entries in place (to some fixed precision), i.e.\
\begin{equation}
\label{eqref10}
  U_A\ket{i}\ket{j}\ket{z} = \ket{i}\ket{j}\ket{z\oplus A_{ij}},
\end{equation}
where $0\le i,j\le n-1$ are the row and column indices, and the third register contains the matrix entry to $p$-bits of precision (so $0\le z\le 2^p$). 

Another oracle that we will use is the comparator
\begin{equation}
    \mathfrak{comp}\ket{a}\ket{b}\ket{0}_{\text{flag}} = \begin{cases}
                                                            \ket{a}\ket{b}\ket{0}_{\text{flag}} & \text{if}~a<b  \\
                                                            \ket{a}\ket{b}\ket{1}_{\text{flag}} & \text{if}~a\ge b,
                                                        \end{cases}
\end{equation}
which can be implemented efficiently using quantum adder circuits \cite{Gidney2017HalvingAddition}.

Let us use the $U_A$ oracle to prepare a superposition over the entries of a chosen row $0\le i\le n-1$
\begin{eqnarray}
    \ket{i}\ket{0}\ket{0}_{\text{data}}&\xrightarrow{\id\otimes H\otimes\id}&\frac{1}{\sqrt{n}}\sum_{j=0}^n\ket{i}\ket{j}\ket{0}_{\text{data}}\nonumber\\
        &\xrightarrow{U_A}&\frac{1}{\sqrt{n}}\sum_{j=0}^n\ket{i}\ket{j}\ket{A_{ij}}_{\text{data}} \nonumber\\
        &=:& \ket{\text{row}_i}
\end{eqnarray}

Then we can adjoin two ancillary registers, $\ket{\delta}_{\text{ref}}\ket{0}_{\text{flag}}$, and using the $\mathfrak{comp}$ oracle we can make the following series of transformations
\begin{align*}
    \ket{\text{row}_i}\ket{\delta}_{\text{ref}}\ket{0}_{\text{flag}}
            \xrightarrow{\id\otimes\id\otimes~\mathfrak{comp}}& \ket{i}\otimes\left(
                            \frac{1}{\sqrt{n}}\sum_{j\in\text{supp}{^\delta}_i(A)}\ket{j}\ket{A_{ij}}_{\text{data}}\ket{\delta}_{\text{ref}}\ket{0}_{\text{flag}}\right. \\
                            &\quad\qquad + \left. \frac{1}{\sqrt{n}}\sum_{j\notin\text{supp}{^\delta}_i(A)}\ket{j}\ket{A_{ij}}_{\text{data}}\ket{\delta}_{\text{ref}}\ket{1}_{\text{flag}}
                            \vphantom{\ket{A_{ij}}}\right), \\
                            & =: \ket{\text{spar}_i}
\end{align*}
where the support of row $i$ of $A$ is $\text{supp}^{\delta}_i(A)=\set{0\le j\le n-1 : |A_{ij}|\geq\delta}$. Here we have assumed 
that the data register contains the magnitude of $A_{ij}$, so that we can just check if it is less than a small threshold $\delta$ in order to check if it is close to zero --- the magnitude can be obtained easily by taking advantage of the signed fixed-point representation of $A_{ij}$ (e.g.\ by simply neglecting the sign bit). Now note that the amplitude of the $\ket{1}_{\text{flag}}$ subspace of the above state is proportional to the sparsity $s(i):=\Abs{\text{supp}^{\delta}_i(A)}$ of row $i$
\begin{align}
    \Norm{\Pi_1^{\text{flag}}\ket{\text{spar}_i}} 
    &= \Norm{\frac{1}{\sqrt{n}}\sum_{j\notin\text{supp}{^\delta}_i(A)}\ket{j}\ket{A_{ij}}_{\text{data}}\ket{\delta}_{\text{ref}}}
    \nonumber\\&=\sqrt{\frac{s(i)}{n}},
\end{align}
where $\Pi_1^{\text{flag}}=\id_r\otimes\id_c\otimes\id_{\text{data}}\otimes\id_{\text{ref}}\otimes\proj{1}_{\text{flag}}$ is a projector onto the flag $1$ subspace. This can be estimated to additive precision $\epsilon'$ using $\Theta\left(1/\epsilon\right)$ queries to $U_A$, using the method of quantum amplitude estimation \cite{Brassard2002QuantumEstimation}, which would give us a quantity $\tilde{x}_i$ satisfying
\[
\Abs{\sqrt{\frac{s(i)}{n}}-\tilde{x}_i}\le \epsilon',
\]
whence we see that choosing $\epsilon'=\epsilon/\sqrt{n}$ gives us an additive approximation of $\sqrt{n}\tilde{x}_i$ of $s(i)$ to precision $\epsilon$, using $\O(\sqrt{n}/\epsilon)$ queries. Following the same procedure to estimate the sparsities of all $n$ rows, the overall sparsity (which for us is the maximum number of non zeros in any row or column) can be ascertained in $\O(n^{3/2})$ queries; we can leave $\epsilon$ out of this consideration since it can be chosen to be of order unity.

\subsection{Sparsity testing using quantum maximum finding}

We still use oracle access as in Eq.~(\ref{eqref10}), and we assume the data register has enough qubits to store the sum of the entries in any row. We start by putting the rows in superposition:

\begin{equation}
    \ket{0}\ket{0}\ket{0}_{\text{data}}\xrightarrow{ H\otimes\id \otimes \id} \frac{1}{\sqrt{n}}\sum_{i=0}^n\ket{i}\ket{0}\ket{0}_{\text{data}} 
\end{equation}

Now we iterate over $n$ calls to the oracle ($j$ is initially 0):
\begin{equation}
    \frac{1}{\sqrt{n}}\sum_{i=0}^{n-1}\ket{i}\ket{j}\ket{\sum_{k=0}^{j-1} A_{ik}}_{\text{data}} \xrightarrow{U_A} \frac{1}{\sqrt{n}}\sum_{i=0}^{n-1}\ket{i}\ket{j}\ket{\sum_{k=0}^{j} A_{ik}}_{\text{data}},
\end{equation}
setting $j \leftarrow j+1$ on each iteration, until $j=n-1$, after which we have the state:
\begin{equation}
    \frac{1}{\sqrt{n}}\sum_{i=0}^{n-1}\ket{i}\ket{n-1}\ket{\sum_{k=0}^{n-1} A_{i,k}}_{\text{data}},
\end{equation}
in which the column register, now in state $\ket{n}$, can be dispensed with. Therefore in $\mathcal{O}(n)$ operations we have created a superposition of the sum of each of the $n$ rows, indexed accordingly. Quantum maximum finding methods \cite{Durr96MaxFind} can make use of this state, preparing it $\O(\sqrt{n})$ times, to find the maximum in a further. Thus we have a quantum algorithm that takes $\mathcal{O}(n^{3/2})$ oracle queries and additional quantum arithmetic operations. By contrast, a classical algorithm to check for row sparsity would have to sum over all rows ($\mathcal{O}(n^2)$ operations) and then classically find the maximum ($\mathcal{O}(n)$ operations). (note that it may be possible to do this slightly faster, but it would still be necessary to check over a number of elements growing linearly with $n$ for each row, and to check all of the $n$ rows).

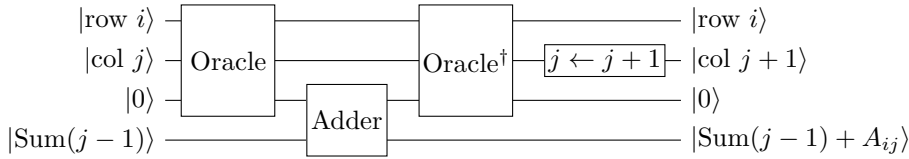
\begin{figure}[!h]
\centering
\begin{tikzpicture}[scale=1.000000,x=1pt,y=1pt]
\filldraw[color=white] (0.000000, -7.500000) rectangle (193.000000, 52.500000);
\draw[color=black] (0.000000,45.000000) -- (193.000000,45.000000);
\draw[color=black] (0.000000,45.000000) node[left] {$|{\text{row}~i}\rangle$};
\draw[color=black] (0.000000,30.000000) -- (193.000000,30.000000);
\draw[color=black] (0.000000,30.000000) node[left] {$|{\text{col}~j}\rangle$};
\draw[color=black] (0.000000,15.000000) -- (193.000000,15.000000);
\draw[color=black] (0.000000,15.000000) node[left] {$|0\rangle$};
\draw[color=black] (0.000000,0.000000) -- (193.000000,0.000000);
\draw[color=black] (0.000000,0.000000) node[left] {$|\text{Sum}(j-1)\rangle$};
\draw (23.500000,45.000000) -- (23.500000,15.000000);
\begin{scope}
\draw[fill=white] (23.500000, 30.000000) +(-45.000000:24.748737pt and 29.698485pt) -- +(45.000000:24.748737pt and 29.698485pt) -- +(135.000000:24.748737pt and 29.698485pt) -- +(225.000000:24.748737pt and 29.698485pt) -- cycle;
\clip (23.500000, 30.000000) +(-45.000000:24.748737pt and 29.698485pt) -- +(45.000000:24.748737pt and 29.698485pt) -- +(135.000000:24.748737pt and 29.698485pt) -- +(225.000000:24.748737pt and 29.698485pt) -- cycle;
\draw (23.500000, 30.000000) node {Oracle};
\end{scope}
\draw (68.000000,15.000000) -- (68.000000,0.000000);
\begin{scope}
\draw[fill=white] (68.000000, 7.500000) +(-45.000000:21.213203pt and 19.091883pt) -- +(45.000000:21.213203pt and 19.091883pt) -- +(135.000000:21.213203pt and 19.091883pt) -- +(225.000000:21.213203pt and 19.091883pt) -- cycle;
\clip (68.000000, 7.500000) +(-45.000000:21.213203pt and 19.091883pt) -- +(45.000000:21.213203pt and 19.091883pt) -- +(135.000000:21.213203pt and 19.091883pt) -- +(225.000000:21.213203pt and 19.091883pt) -- cycle;
\draw (68.000000, 7.500000) node {Adder};
\end{scope}
\draw (112.500000,45.000000) -- (112.500000,15.000000);
\begin{scope}
\draw[fill=white] (112.500000, 30.000000) +(-45.000000:24.748737pt and 29.698485pt) -- +(45.000000:24.748737pt and 29.698485pt) -- +(135.000000:24.748737pt and 29.698485pt) -- +(225.000000:24.748737pt and 29.698485pt) -- cycle;
\clip (112.500000, 30.000000) +(-45.000000:24.748737pt and 29.698485pt) -- +(45.000000:24.748737pt and 29.698485pt) -- +(135.000000:24.748737pt and 29.698485pt) -- +(225.000000:24.748737pt and 29.698485pt) -- cycle;
\draw (112.500000, 30.000000) node {Oracle$^\dagger$};
\end{scope}
\begin{scope}
\draw[fill=white] (164.500000, 30.000000) +(-45.000000:31.819805pt and 8.485281pt) -- +(45.000000:31.819805pt and 8.485281pt) -- +(135.000000:31.819805pt and 8.485281pt) -- +(225.000000:31.819805pt and 8.485281pt) -- cycle;
\clip (164.500000, 30.000000) +(-45.000000:31.819805pt and 8.485281pt) -- +(45.000000:31.819805pt and 8.485281pt) -- +(135.000000:31.819805pt and 8.485281pt) -- +(225.000000:31.819805pt and 8.485281pt) -- cycle;
\draw (164.500000, 30.000000) node {$j\leftarrow j+1$};
\end{scope}
\draw[color=black] (193.000000,45.000000) node[right] {$|{\text{row}~i}\rangle$};
\draw[color=black] (193.000000,30.000000) node[right] {$|\text{col}~j+1\rangle$};
\draw[color=black] (193.000000,15.000000) node[right] {$|0\rangle$};
\draw[color=black] (193.000000,0.000000) node[right] {$|\text{Sum}(j-1)+A_{ij}\rangle$};
\end{tikzpicture}
\captionsetup{width=0.95\linewidth}
\caption{\small{Circuit portion showing how the sum of $A_{ij}$ can be obtained inside the ket. The oracle loads the matrix entries $A_{ij}$ into the third register, and $\text{Sum}(j)=\sum_{k=0}^j A_{ij}$.}}
\label{f:circuit}
\end{figure}

We remark that the above algorithm that uses quantum maximum finding appears to rely on $A$ being a binary adjacency matrix. When an upper bound $\Lambda$ on $\Norm{A}_{max}$ is available, this limitation can be overcome by normalising the matrix entries by $\Lambda$.






\section{Discussion}
\label{disc}

In this paper we have shown how classical sparsifying techniques can be used as a preprocessing step to obtain a row computable sparse input matrix that can then be used with efficient quantum algorithms for sparse Hamiltonian simulation. The one-time $\O(\frac{m\log n}{\epsilon^2})$ classical overhead in runtime may be justified by the fact that the sparsified output matrix may be used for multiple applications (e.g.\ for simulating time evolution of several different states) which can be efficiently performed quantumly.

Usually we require every problem instance to be row-sparse in a quantum algorithm. What we have so far, using spectral sparsification, is a guarantee that as $n$ grows, the sparsified output is also row sparse with high probability. Therefore it is also necessary that the simulation has some sort of checking mechanism, such that the simulation is halted if too many iterations have been required (i.e., because the actual sparsified Hamiltonian that was generated was not, in fact, row sparse), and to start again with a fresh sparsified Hamiltonian. This should be easy to include in any implementation, and the first and second properties of row sparsity are sufficient in this case to guarantee good overall performance (that is, as $n \to \infty$ the probability of needing to start again vanishes).\\
\indent On a more general note, it is interesting to consider Hamiltonian sparsification in the context of a suite of simulation algorithms. For example, we have identified that star graphs are sparse, but not row sparse -- thus we can see that physical systems that are dominated by a few components may yield sparsified Hamiltonian's that are essentially a superposition of a number of star graphs. Thus, whilst the techniques presented in this paper will not apply, it may be possible to use other techniques such as low-rank approximations. Conversely, for physical systems in which a number of components barely make any impact on the whole (i.e., they have few and/or low-weight edges to other vertices), then it is likely to be safe to simply neglect these components. Informally, this can be seen as a justification for assuming that each row was sampled at least $\textnormal{poly} \log$ many times, as in Section~\ref{lap}.
\subsection*{Open problems}
Finally, it is worth discussing the deficiencies of this paper -- as identified in the introduction, we make three assumptions in the analysis: that the Hamiltonian is real, that its rows are all sampled at least $\textnormal{poly} \log$ many times, and that it commutes with its sparsifier. The first two of these essentially restrict the physical application of our work, and it would therefore be beneficial to show that the same results hold when these assumptions are removed, as well as tightening the various bounds where possible; the third assumption, however, is more fundamental --  it is important to understand whether a sparsified Hamiltonian commutes with the actual Hamiltonian in general, and if not whether the discrepancy can be shown to be insignificant when the full simulation is analysed (for example by using Trotter formulas / the BCH expansion to quantify errors). However, such a question seems to be of more general relevance than simply to plug a gap in our analysis: the condition given in Eq.~(\ref{eq:spectral_approx_H}) seems to be an eminently reasonable general measure of approximation accuracy, which may be used for myriad Hamiltonian approximation methods, and it is therefore important to show that it does indeed lead to accurate Hamiltonian simulation.

    


\clearpage
\bibliographystyle{unsrt}





\end{document}